\documentclass[11pt]{article}
\usepackage{amsmath}
\usepackage{amssymb}
\usepackage{amsfonts}
\usepackage{latexsym}
\usepackage{color}
\newcommand{\be}{\begin{equation}}
\newcommand{\en}{\end{equation}}
\newcommand{\mc}{\mathcal}
\newcommand{\mb}{\mathbb}
\newcommand{\D}{\mc D}
\newcommand{\Sys}{{\cal S}}

\newcommand{\Hil}{\mc H}
\newcommand{\A}{\mc A}
\newcommand{\Ao}{{{\mc A}_0}}
\newcommand{\Bo}{{{\mc B}_0}}
\newcommand{\LL}{\mc L}

\newcommand{\LD}{{\LL}^\dagger (\D)}
\newcommand{\1}{1\!\!\!\!1}
\newcommand{\LDD}{{\LL} (\D_\pi,\D_\pi')}
\newcommand{\LDH}{{\LL}^\dagger (\D,\Hil)}

\newtheorem{defn}{Definition}[section]
\newtheorem{prop}[defn]{Proposition}
\newtheorem{thm}[defn]{Theorem}
\newtheorem{lemma}[defn]{Lemma}
\newtheorem{cor}[defn]{Corollary}
\newenvironment{proof}{\noindent {\bf Proof --}}{\hfill$\square$ \vspace{3mm}\endtrivlist}

\begin{document}

\thispagestyle{empty}

\vspace*{0.7cm}

\begin{center}
{\Large \bf Exponentiating derivations of quasi *-algebras:
\\possible approaches and applications}

\vspace{1.5cm}

{\large F. Bagarello }
\vspace{3mm}\\
 Dipartimento di Metodi e Modelli Matematici \\
Facolt\`a d'Ingegneria - Universit\`a di Palermo \\ Viale delle
Scienze,       \baselineskip15pt
     I-90128 - Palermo - Italy, \vspace{2mm}\\

\vspace{5mm}

{\large A. Inoue}
\vspace{3mm}\\
Department of Applied Mathematics,\\ Fukuoka University,\\
J-814-80 Fukuoka, Japan

\vspace{3mm} and

\vspace{3mm} {\large C.Trapani}
\vspace{3mm}\\
 Dipartimento di Matematica e Applicazioni \\
Universit\`a di Palermo\\Via Archirafi 34, \baselineskip15pt
I-90123 - Palermo - Italy\vspace{2mm}

\end{center}

\vspace*{12mm}

 \noindent{\sc Abstract} \\\noindent {The problem of exponentiating derivations
 of quasi *-algebras is considered in view of applying it to the determination of
  the time evolution of a physical system. The particular case where
 observables constitute a proper CQ*-algebra is analyzed.}

\vspace{3cm} \noindent{\em 2000 Mathematics Subject
Classification}: 47L60; 47L90.
 \vfill

\newpage

\section{Introduction}
 The unbounded nature of the operators describing
observables of a quantum mechanical system with a finite or an
infinite number of degrees of freedom is mathematically {\em a
fact} which follows directly from the non commutative nature of
the quantum world in the sense that, as a consequence of the
Wiener-von Neumann theorem, the commutation relation $[\hat q,\hat
p]=i\1$ for the position $\hat q$ and the momentum $\hat p$ is not
compatible with the boundedness of both $\hat q$ and $\hat p$.
Thus any operator representation of this commutation relation
necessarily involves unbounded operators. Also the bosonic
creation and annihilation operators $a^\dagger$ and $a$,
$[a,a^\dagger]=\1$, or  the hamiltonian of the simple harmonic
oscillator, $H=\frac{1}{2}(\hat p^2+\hat q^2)=a^\dagger
a+\frac{1}{2}\1$, just to mention few examples, are all unbounded
operators.

However, when an experiment is carried out, what is measured is an
eigenvalue of an observable, which is surely a \underline{finite}
real number: for instance, if the physical system $\Sys$ on which
 measurements are performed is in a laboratory, then if we measure
the position of a particle of $\Sys$ we must get a finite number
as a result. Also, if we measure the energy of a quantum particle
in a, say, harmonic potential, we can only get a finite measure
since the probability that the particle has infinite energy is
zero. Moreover, in a true relativistic world, since the velocity
of a particle cannot exceed the velocity of  light $c$, any
measurement of its momentum can only give, again, a finite result.
>From the mathematical point of view this may correspond to
restricting the operator to some {\it spectral subspaces} where
the unboundedness is in fact removed. This procedure supports the
practical point of view  where it seems enough to deal, from the
very beginning, with bounded operators only.

It is then reasonable to look for a compromise within these
opposite approaches and the compromise could be the following:
given a system $\Sys$ we consider a slightly modified version of
it in which all the operators related to $\Sys$ are replaced by
their {\em regularized} version (e.g. their {\em finite-volume}
version, natural choice if $\Sys$ {\em lives} in a laboratory!),
obtained by means of some given cutoff, we compute all those
quantities which are relevant for our purposes and then, in order
to check whether this procedure has modified the original physical
nature of $\Sys$, we try to see whether these results are {\em
stable} under the removal of the cutoff. As an example, if $\Sys$
is contained inside a box of volume $V$, we do expect that all the
results become independent of the volume cutoff as soon as the
value of this cutoff, $W$, becomes larger than $V$, since what
happens outside the box has almost no role in the behavior of
$\Sys$.

As it is extensively discussed in \cite{sew}, the full description
of a physical system $\Sys$ implies the knowledge of three basic
ingredients: the set of the observables, the set of the states
and, finally, the dynamics that describes the time evolution of
the system by means of the time dependence of the expectation
value of a given observable on a given state. Originally the set
of the observables was considered to be a C*-algebra, \cite{haag}.
In many applications, however, this was shown not to be the most
convenient choice and the C*-algebra was replaced by a von Neumann
algebra, because the role of the representations turns out to be
crucial mainly when long range interactions are involved, see
\cite{bm} and references therein. Here we use a different
algebraic structure, similar to the one considered in \cite{eka},
which is suggested by the considerations above: because of the
relevance of the unbounded operators in the description of $\Sys$,
we will assume is Sections 2 and 3 that the observables of the
system belong to a quasi *-algebra $(\A,\Ao)$, see \cite{ctrev}
and references therein, while,  in order to have a richer
mathematical structure, in Section 4 we will use a slightly
different algebraic structure: $(\A,\Ao)$ will be assumed to be a
proper CQ*-algebra, which has nicer topological properties. In
particular, for instance, $\Ao$ is a C*-algebra. The set of states
over $(\A,\Ao)$, $\Sigma$, is described again in \cite{ctrev},
while the dynamics is usually a group (or a semigroup) of
automorphisms of the algebra, $\alpha^t$. Therefore, following
\cite{sew}, we simply put $\Sys=\{(\A,\Ao),\Sigma,\alpha^t\}$.

The system $\Sys$ is now {\em regularized}: we introduce some
cutoff $L$, (e.g. a volume or an occupation number cutoff),
belonging to a certain set $\Lambda$, so that $\Sys$ is replaced
by a sequence or, more generally, a net of systems $\Sys_L$, one
for each value of $L\in\Lambda$. This cutoff is chosen in such a
way that all the observables of $\Sys_L$ belong to a certain
*-algebra $\A_L$ contained in $\Ao$: $\A_L\subset\Ao\subset\A$. As
for the states, we choose $\Sigma_L=\Sigma$, that is, the set of
states over $\A_L$ is taken to coincide with the set of states
over $\A$. This is a common choice, \cite{bm}, even if also
different possibilities are considered in literature. For
instance, in \cite{sewbag}, also the states  depend on $L$.
Finally, since the dynamics is related to a hamiltonian operator
$H$ (or to the Lindblad generator of a semigroup), and since $H$
has to be replaced with $H_L$, because of the cutoff, $\alpha^t$
is replaced by the family $\alpha_L^t(\cdot)=e^{iH_Lt}\cdot
e^{-iH_Lt}$. Therefore
$$
\Sys=\{(\A,\Ao),\Sigma,\alpha^t\}\longrightarrow\{\Sys_L=\{\A_L,\Sigma,\alpha_L^t\},L\in\Lambda\}.
$$

 \section{The mathematical framework}

Let $\A$ be a linear space and $\Ao$ a  $^\ast$-algebra contained
in $\A$ as a subspace. We say that $\A$ is a quasi $^\ast$-algebra
over $\Ao$ if (i) the right and left multiplications of an element
of $\A$ and an element of $\Ao$ are always defined and linear;
(ii) $x_1 (x_2 a)= (x_1x_2 )a, (ax_1)x_2= a(x_1 x_2)$ and $x_1(a
 x_2)= (x_1 a) x_2$, for each $x_1, x_2 \in \A_0$ and $a \in \A$ and (iii) an
involution * (which extends the involution of $\Ao$) is defined in
$\A$ with the property $(ab) ^\ast =b ^\ast a ^\ast$ whenever the
multiplication is defined.

 In this paper we will always assume
that the quasi $^\ast$ -algebra under consideration has a unit,
i.e. an element $\1 \in \Ao$ such that $a\1 =\1 a=a, \;\, \forall
a\in \A$.

A quasi  $^\ast$ -algebra $(\A,\Ao)$ is said to be a locally
convex quasi
 $^\ast$-algebra if in $\A$ a locally convex topology $\tau$ is defined  such
that (a) the involution is continuous and the multiplications are
separately continuous; and (b) $\Ao$ is dense in $\cal A[\tau]$.
We indicate with $\{p_\alpha\}$ a directed set of seminorms which
defines $\tau$. { Throughout this paper, we will always suppose,
without loss of generality, that a locally convex quasi *-algebra
$(\A[\tau],\Ao)$ is {\it complete}.}

In the following we also need the concept of {\em
*-representation}.

Let $\D$ be a dense subspace in some Hilbert space $\Hil$. We
denote with $\LDH$ the set of all closable operators $X$ in $\Hil$
such that
 $D(X)  = \D$ and $D(X^*)  \supset  \D$ which is a {\em partial
 *-algebra}, \cite{book},
 with the usual operations
 $X+Y$, $\lambda X$, the involution $X^\dagger=X^*|\D$ and the weak product
 $X${\tiny$\Box$}$Y\equiv X^{\dagger *}Y$ whenever $Y\D\subset D(X^{\dagger *})$ and $X^\dagger\D\subset D(Y^{*})$.
We also denote with $\LD$ the *-algebra consisting of the elements
$A\in \LDH$ such that both $A$ and its adjoint $A^*$ map $\D$ into
itself (in this case, the weak multiplication reduces to the
ordinary multiplication of operators).

\smallskip Let $(\A,\Ao)$ be a quasi *-algebra, $\D_\pi$ a dense
domain in a certain Hilbert space $\Hil_\pi$, and $\pi$ a linear
map from $\A$ into $\LL^\dagger(\D_\pi, \Hil_\pi)$ such that:

(i) $\pi(a^*)=\pi(a)^\dagger, \quad \forall a\in \A$;

(ii) if $a\in \A$, $x\in \Ao$, then $\pi(a)${\tiny$\Box$}\!\!
$\pi(x)$ is well defined and $\pi(ax)=\pi(a)${\tiny$\Box$}\!\!
$\pi(x)$.

We say that such a map $\pi$ is a {\em *-representation of $\A$}.
Moreover, if

(iii) $\pi(\Ao)\subset \LL^\dagger(\D_\pi)$,

then $\pi$ is said to be a {\em *-representation of the quasi
*-algebra} $(\A,\Ao)$.

{ The *-representation $\pi$ is called {\it ultra-cyclic} if there
exists $\xi_0 \in \D_\pi$ such that $\pi(\Ao)\xi_0=\D_\pi$.}

  Let $\pi$ be a *-representation of $\A$.
The strong topology $\tau_s$ on $\pi(\A)$ is the locally convex
topology defined by the following family of seminorms:
$\{p_\xi(.); \; \xi\in\D_\pi\}$, where
$p_\xi(\pi(a))\equiv\|\pi(a)\xi\|$, where $a\in \A$, $\xi\in
\D_\pi$.

For an overview on partial *-algebras and related topics we refer
to \cite{book}.

\begin{defn}
Let $(\A,\Ao)$ be a quasi *-algebra. A {\em *-derivation of}
$\Ao$ is a map $\delta: \Ao\rightarrow \A$
 with the following properties:
\begin{itemize}
\item[(i)]  $\delta(x^*)=\delta(x)^*, \; \forall x \in \Ao$;
\item[(ii)] $\delta(\alpha x+\beta y) = \alpha \delta(
x)+\beta\delta( y), \; \forall x,y \in \Ao, \forall \alpha,\beta
\in \mathbb{C}$; \item [(iii)] $\delta(xy) = x\delta( y)+\delta(
x)y,  \; \forall x,y \in \Ao$.
\end{itemize}
\label{Definition 3.1}
\end{defn}

Let $(\A,\Ao)$ be a quasi *-algebra and $\delta$ be a *-derivation
of $\Ao$. Let $\pi$ be a *-representation of  $(\A,\Ao)$. We will
always assume that whenever $x\in \Ao$ is such that $\pi(x)=0$,
then $\pi(\delta(x))=0$. Under this assumption, the linear map
\begin{equation}
\delta_\pi(\pi(x))=\pi(\delta(x)), \quad x\in \Ao, \label{41}
\end{equation}
is well-defined on $\pi(\Ao)$ with values in $\pi(\A)$ and it is a
*-derivation of $\pi(\Ao)$. We call $\delta_\pi$ the *-derivation
{\em induced} by $\pi$.

Given such a representation $\pi$ and its dense domain $\D_\pi$,
we consider the usual graph topology $t_\dagger$ generated by the
seminorms

\begin{equation}
\xi\in\D_\pi \rightarrow \|A\xi\|, \quad A\in \LL^\dagger(\D_\pi).
\label{42}
\end{equation}

{ If $\D_\pi'$ denotes the conjugate dual space of $\D_\pi$, we
get the usual rigged Hilbert space $\D_\pi[t_\dagger] \subset
\Hil_\pi \subset \D_\pi'[t_\dagger']$, where $t_\dagger'$ denotes
the strong dual topology of $\D_\pi'$.  Let $\LDD$ denote the
space of all continuous linear maps from $\D_\pi[t_\dagger]$ into
$\D_\pi'[t_\dagger']$. Then one has
$$
\LL^\dagger(\D_\pi) \;  \subset \; \LDD.$$} Each operator $A\in
\LL^\dagger(\D_\pi)$ can be extended to an operator $\hat A$ on
the whole $\D_\pi'$ in the following way:
\begin{equation}\label{extens}
<\hat A\xi',\eta>=<\xi',A^\dagger \eta>, \quad \forall \xi'\in
\D_\pi', \quad \eta\in \D_\pi.
\end{equation}
Therefore the left and right multiplication of
$X\in\LL(\D_\pi,\D_\pi')$ and $A\in\LL^\dagger(\D_\pi)$ can always
be defined:
$$
(X\circ A)\xi=X(A\xi), \mbox{ and } (A\circ X)\xi=\hat A(X\xi),
\quad \forall \xi\in \D_\pi.
$$

\medskip With these definitions it is known that
$(\LL(\D_\pi,\D_\pi'),\LL^\dagger(\D_\pi))$ is a quasi *-algebra.

{Let $\delta$ be a *-derivation of $\Ao$ and $\pi$ a
*-representation of $(\A,\Ao)$. Then
$\pi(\Ao)\subset\LL^\dagger(\D_\pi)$. We say that the *-derivation
$\delta_\pi$ induced by  $\pi$  is {\em spatial} if there exists
$H=H^\dagger\in \LL(\D_\pi,\D_\pi')$ such that
\begin{equation}
\delta_\pi(\pi(x))=i\{H\circ\pi(x)-\widehat{\pi(x)}\circ H\},\quad
\forall x\in\Ao. \label{43}
\end{equation}
where $\widehat{\pi(x)}$ denotes the extension of $\pi(x)$ defined
as in \eqref{extens} (from now on, whenever no confusion may
arise, we use the same notation for $\pi(x)$ and for its
extension).

\medskip
Let now $(\A,\Ao)$ be a locally convex quasi *-algebra with
topology $\tau$. Necessary and sufficient conditions for the
existence of a $(\tau-\tau_s)$-continuous, ultra-cyclic
*-representation $\pi$ of $\A$, with ultra-cyclic vector $\xi_0$,
such that the *-derivation $\delta_\pi$ induced by $\pi$  is
spatial have been given in \cite[Theorem 4.1]{bit}. We now suppose
that these conditions occur, so that there exists a ultra-cyclic
$(\tau-\tau_s)$-continuous *-representation $\pi$ of $\A$ in
Hilbert space $\Hil_\pi$,  with ultra-cyclic vector $\xi_0$.
Furhermore, we assume that a family of *-derivations (in the sense
of Definition \ref{Definition 3.1}) $\{\delta_n: \, n\in \mathbb
N\}$ of the
 *-algebra with identity $\Ao$ is given. As done in \cite{bit}, we consider
 the related family of *-derivations $\delta_\pi^{(n)}$
induced by $\pi$ defined on $\pi(\Ao)$ and with values in
$\pi(\A)$:
\begin{equation}
 \delta_\pi^{(n)}(\pi(x))=\pi(\delta_n(x)), \quad x\in\Ao.
\label{417}
\end{equation}
Suppose that each $\delta_\pi^{(n)}$ is spatial and let $H_n\in
\LDD$ be the corresponding implementing operator. Assume, in
addition, that
$$\sup_n\|H_n\xi_0\|=:L<\infty.$$ Then, as shown in \cite[Proposition 4.3]{bit},
 if $\{\delta_n(x)\}$
$\tau$-converges to $\delta (x)$, for every $x \in \Ao$, it turns
out that $\delta$ is a *-derivation of $\Ao$ and the *-derivation
$\delta_\pi$ induced by $\pi$ is well-defined and spatial. The
relation between $H_n$ and the operator $H$ implementing
$\delta_\pi$ has also been discussed.

\vspace{2mm} The above statements reveal to be crucial for the
discussion of the existence of the dynamics of systems where a
{\it cut-off} has been introduced, as we shall see later. Examples
for which these conditions are satisfied have been discussed in
\cite{bit}(Examples 4.4 and 4.5).

\section{Applications to regularized systems}

As we have discussed in the Introduction, given a physical system
$\Sys$, the first step in dealing with it consists in replacing
$\Sys$ with a whole family of {\em regularized} systems
$\{\Sys_L=\{\A_L,\Sigma,\alpha_L^t\}, L\in\Lambda\}$, obtained by
introducing some cutoff which is related to $\Sys$ itself. We
suppose that the dynamics $\alpha_L^t$ is generated by a
*-derivation $\delta_L$. The procedure of the previous section
suggests to introduce the following

\begin{defn}
A family $\{\Sys_L, L\in\Lambda\}$  is said to be c-representable
if  there exists a *-representation $\pi$ of $(\A,\Ao)$ such that:
\begin{itemize}
\item[(i)]  $\pi$ is $(\tau-\tau_s)$-continuous; \item[(ii)] $\pi$
is ultra-cyclic with ultra-cyclic vector $\xi_0$; \item [(iii)] if
$\pi$ is such that $\pi(x)=0$, then $\pi(\delta_L(x))=0$, $\forall
L\in\Lambda$.

Any such representation $\pi$ is said to be a c-representation.
\end{itemize}
\label{Definition 41}
\end{defn}

\begin{prop}

Let $\{\Sys_L, L\in\Lambda\}$ be a c-representable family and
$\pi$ a c-representation. Let $h_L=h_L^*\in\A_L$ be the element
which implements $\delta_L$: $\delta_L(x)=i[h_L,x]$, for all
$x\in\Ao$. Suppose that the following conditions are satisfied:

(1) $\delta_L(x)$ is $\tau$-Cauchy for all $x\in\Ao$;

(2) $\sup_L\|\pi(h_L)\xi_0\|<\infty$.

Then, one has
\begin{itemize}
\item[(a)] $\delta(x)=\tau-\lim_L\delta_L(x)$ exists in $\A$ and
is a *-derivation of $\Ao$; \item[(b)] $\delta_\pi$, the
*-derivation induced by $\pi$, is well defined and spatial.
\end{itemize}
\label{theorem41}
\end{prop}
\begin{proof}
The proof of the first statement is trivial.

Let us define $\delta_L^{(\pi)}(\pi(x))=\pi(\delta_L(x))$, $x\in
\Ao$, and $H_L^{(\pi)}=\pi(h_L)$. Then we have
$\delta_L^{(\pi)}(\pi(x))=i[h_L^{(\pi)},\pi(x)]$, which means that
$\delta_L^{(\pi)}(\pi(x))$ is spatial and it is implemented by
$H_L^{(\pi)}$. In order to apply Proposition 4.3 of \cite{bit} we
have to check that $H_L^{(\pi)}$ satisfies the following
requirements: (a) $H_L^{(\pi)}={H_L^{(\pi)}}^\dagger$; (b)
$H_L^{(\pi)}\in\LL(\D_\pi,\D_\pi')$; (c)
$H_L^{(\pi)}\xi_0\in\Hil_\pi$; (d)
$\delta_L^{(\pi)}(\pi(x))=i\{H_L^{(\pi)}\circ\pi(x)-\pi(x)\circ
H_L^{(\pi)}\}$; (e) $\sup_L\|\pi(H_L)\xi_0\|<\infty$.

Condition (a) follows from the self-adjointness of $h_L$ and from
the fact that $\pi$ is a *-representation. Condition (b) holds in
an even stronger form. In fact, since $h_L$ belongs to
$\A_L\subset\Ao$, then
$H_L^{(\pi)}\in\LL^\dagger(\D_\pi)\subset\LL(\D_\pi,\D_\pi')$. For
this reason we also have that
$H_L^{(\pi)}\xi_0\in\D_\pi\subset\Hil_\pi$ while condition (d) is
satisfied without even the need of using the "$\circ$"
multiplication. Finally, condition (e) coincides with assumption
(2). Proposition 4.3 of \cite{bit} implies therefore the statement
and, in particular, it says that the implementing operator of
$\delta^{(\pi)}$, $H^{(\pi)}$, satisfies the following properties:
$H^{(\pi)}={H^{(\pi)}}^\dagger$;
$H^{(\pi)}\in\LL(\D_\pi,\D_\pi')$; $H^{(\pi)}\xi_0\in\Hil_\pi$ and
$\delta^{(\pi)}(\pi(x))=i\{H^{(\pi)}\circ\pi(x)-\pi(x)\circ
H^{(\pi)}\}$, $\forall x\in\Ao$.

\end{proof}

\vspace{2mm}

{\bf Remarks:--} (1) It is clear that if the sequence $\{h_L\}$ is
$\tau$-convergent, then assumption (2) of the above Proposition is
automatically satisfied, at least if $L$ is a discrete index.

(2) It is interesting to observe also that the outcome of this
Proposition is that any physical system $\Sys$ whose related
family $\{\Sys_L, L\in\Lambda\}$ is c-representable {\em admits an
effective hamiltonian in the sense of } \cite{bagtra1996}.

\vspace{4mm}

Now we show how to use the previous results, together with some
statement contained in \cite{bagtra1996}, to define the time
evolution of $\Sys$. We will use here quite a special strategy,
which is suggested by our previous result on the existence of an
effective hamiltonian. Many other possibilities could be
considered as well, and we will discuss some of them in Section 4.

\vspace{2mm}

First of all we will assume that the $h_L$'s introduced in the
previous section can be written in terms of some (intensive)
elements $s_L^\alpha$, $\alpha=1,2,..,N$, which are assumed to be
hermitian (this is not a big constraint, of course), and
$\tau$-converging to some elements $s^\alpha\in\A$ { commuting
with all elements of $\Ao$}:
$$
s^\alpha=\tau-\lim_Ls_L^\alpha, \hspace{1cm} [s^\alpha,x]=0,\,
\forall x\in\Ao.
$$
It is worth remarking here that this is what happens, for
instance, in all mean field spin models, where the elements
$s_L^\alpha$ are nothing but the mean magnetization
$s_V^\alpha=\frac{1}{|V|}\sum_{i\in V}\sigma_i^\alpha$, \cite{bm}.

In order to ensure that all the powers of these elements converge,
which is what happens in many concrete applications, \cite{bm},
\cite{bagtra1996} and references therein, we  introduce here the
following definition, which is suggested by \cite{bit}:

\begin{defn}
We say that $\{s_L^\alpha\}$  is a uniformly $\tau$-continuous
sequence if, for each continuous seminorm $p$ of $\tau$ and for
all $\alpha=1,2,...,N$, there exist another continuous seminorm
$q$ of $\tau$ and a positive constant $c_{p,q,\alpha}$ such that
\be p(s_L^\alpha a)\leq c_{p,q,\alpha}q(a), \, \forall a\in\A,
\,\forall L\in\Lambda. \label{51}\en \label{Definition 42}
\end{defn}

Because of the properties of $\tau$ it is easily checked that
(\ref{51}) also implies that $p(as_L^\alpha)\leq
c_{p,q,\alpha}q(a)$, $\forall a\in\A$, and that the same
inequalities can be extended to $s^\alpha$.

It is now straightforward to prove the following

\begin{lemma}
If $\{s_L^\alpha\}$  is a uniformly $\tau$-continuous sequence and
if $\tau-\lim_Ls_L^\alpha=s^\alpha$ for $\alpha=1,2,..,N$, then
$\tau-\lim_L\left(s_L^\alpha\right)^k=\left(s^\alpha\right)^k$ for
$\alpha=1,2,..,N$ and $k=1,2,...$.
\end{lemma}
This lemma has the following consequence: if we define the
multiple commutators $[x,y]_k$ as usual ($[x,y]_1=[x,y]$,
$[x,y]_k=[x,[x,y]_{k-1}]$), then one has
\begin{prop}
Suppose that
\begin{itemize}
\item[(1)] $\forall x\in\Ao$ $[h_L,x]$ depends on $L$ only through
$s_L^\alpha$; \item[(2)]
$s_L^\alpha\stackrel{\tau}{\longrightarrow} s^\alpha$ and
$\{s_L^\alpha\}$  is a uniformly $\tau$-continuous sequence.
\end{itemize}
Then, { for each $k\in \mathbb{N}$,}  the following limit exists
$$
\tau-\lim_Li^k[h_L,x]_k=\tau-\lim_L\delta_L^k(x), \, \forall
x\in\Ao,$$ and defines an element of $\A$ which we call
$\delta^{(k)}(x)$.
\end{prop}

\vspace{2mm}

{\bf Remarks:--} (1) The proof is an easy extension of that given
in \cite{bagtra1996} and will be omitted.

(2) Of course, we could replace condition (2) above directly with
the requirement that the following limits exist:
$\tau-\lim_L\left(s_L^\alpha\right)^k=\left(s^\alpha\right)^k$ for
$\alpha=1,2,..,N$ and $k=1,2,...$.

(3) It is worth noticing that we have used the notation
$\delta^{(k)}(x)$ instead of the more natural $\delta^k(x)$ since
this last quantity could not be well defined because of domain
problems, since we are not working with algebras, in general. In
other words, we cannot claim that
$\pi\left(\tau-\lim_L\delta_L^k(x)\right)=[H^{(\pi)},\pi(x)]_k$,
since the rhs could be not well-defined.

\vspace{2mm}

In order to go on, it is convenient to introduce the following
definition, \cite{bagtra1996}:

\begin{defn}
We say that $x\in \Ao$ is a generalized analytic element of
$\delta$ if, for all $t$, the series
$\sum_{k=0}^\infty\frac{t^k}{k!}\pi(\delta^{(k)}(x))$ is
$\tau_s$-convergent. The set of all generalized elements is
denoted with ${\cal G}$. \label{Definition 43}
\end{defn}

We can now prove  the following

\begin{prop}
Let $x_\gamma$ be a net of elements of $\Ao$ and suppose that,
whenever $\pi(x_\gamma)\stackrel{\tau_s}{\longrightarrow}\pi(x)$
then $x_\gamma\stackrel{\tau}{\rightarrow}x$. Then, $\forall x\in
{\cal G}$ and $\forall t\in \mathbb{R}$, the series
$\sum_{k=0}^\infty\frac{t^k}{k!}\delta^{(k)}(x)$ converges in the
$\tau$-topology to an element of $\A$ which we call $\alpha^t(x)$.

Moreover, $\alpha^t$ can be extended to the $\tau$-closure
$\overline{\cal G}$ of ${\cal G}$. \label{theorem42}
\end{prop}
\begin{proof}
Because of the assumption, given a seminorm $p$ of $\tau$, there
exist a positive constant $c'$ and some vectors $\{\eta_j,\,
j=1,..,n\}$ in $\D_\pi$ such that $p(x)\leq c' \sum_{j=1}^n
p_{\eta_j}(\pi(x))$, for all $x\in\A$.  Then we have, if
$x\in{\cal G}$, $L,M\in\mathbb{N}$ with $M>L$:
$$
p\left(\sum_{k=L}^M\frac{t^k}{k!}\delta^{(k)}(x)\right)\leq c'
\sum_{j=1}^n
p_{\eta_j}\left(\sum_{k=L}^M\frac{t^k}{k!}\pi(\delta^{(k)}(x))\right)\stackrel{L,M}{\longrightarrow}
0,
$$
because $x\in{\cal G}$. Therefore we can put
$\alpha^t(x)=\tau-\sum_{k=0}^\infty\frac{t^k}{k!}\delta^{(k)}(x)$.

The extension of $\alpha^t$ to $\overline{\cal G}$ is simply a
consequence of its $\tau$-continuity.

\end{proof}

{\bf Remark:--} It is worth remarking that the assumptions of this
Proposition are rather strong. In particular, for instance, the
fact that for all $x\in\Ao$ the following estimate holds,
$p(x)\leq c \sum_{j=1}^n p_{\eta_j}(\pi(x))$, implies that the
representation $\pi$ is faithful and that $\pi^{-1}$ is
continuous. Moreover,  the non triviality of the set ${\cal G}$
must be proven case by case. It is also in view of these facts
that in the next section we further specify our algebraic set-up
in order to avoid the use of these strong assumptions in the
analysis of the existence of $\alpha^t$.

\medskip

 We will now discuss an approach, different from the one considered
 so far, to what we have called {\em the exponentiation problem}, i.e.
the possibility of deducing the existence of the time evolution
for certain elements of the *-algebra $\Ao$ (actually a C*-algebra
in many applications) starting from some given *-derivation. In
what follows $\pi$ is assumed to be a faithful *-representation of
the quasi *-algebra $(\A,\Ao)$ and $\delta$  a *-derivation on
$\Ao$. As always we will assume that the *-derivation induced by
$\pi$, $\delta_\pi$, is well-defined on $\pi(\Ao)$ with values in
$\pi(\A)$.

We define the following subset of $\Ao$ \be
\Ao(\delta):=\{x\in\Ao:\quad\delta^k(x)\in\Ao,\quad \forall
k\in\mathbb{N}_0\}.\label{extra1}\en It is clear that
$\Ao(\delta)$ depends on $\delta$: the more regular $\delta$ is,
the larger the set $\Ao(\delta)$ turns out to be. For example, if
$\delta$ is inner in $\Ao$ and the implementing element $h$
belongs to $\Ao$, then $\Ao(\delta)=\Ao$. For general $\delta$, we
can surely say that $\Ao(\delta)$ is not empty since it contains,
at least, all the multiples of the identity $\1$ of $\Ao$.

It is straightforward to check that $\Ao(\delta)$ is a *-algebra
which is mapped into itself by $\delta$. Moreover we also find
that $\pi(\delta^k(x))=\delta_\pi^k(\pi(x))$, for all
$x\in\Ao(\delta)$ and for all $k\in\mathbb{N}_0$. This also
implies that, for all $k\in\mathbb{N}_0$, and for all
$x\in\Ao(\delta)$, $\delta_\pi^k(\pi(x))\in\pi(\Ao)$. This
suggests to introduce the following subset of $\pi(\Ao)$, $
\Ao(\delta)^\pi:=\{\pi(x)\in\pi(\Ao):\,\delta_\pi^k(\pi(x))\in\pi(\Ao),\,
\forall k\in\mathbb{N}_0\}$, and it is clear that
$x\in\Ao\Leftrightarrow \pi(x)\in\pi(\Ao)$.

Let us now introduce on $\A$ the topology $\sigma_s$ defined via
$\tau_s$ in the following way: \be \A\ni a\rightarrow
q_\xi(a)=p_\xi(\pi(a))=\|\pi(a)\xi\|, \quad \xi\in\D_\pi.
\label{extra2}\en {It is worth noticing that $\sigma_s$ does not
make of $(\A, \Ao)$ a locally convex quasi *-algebra, since the
multiplication is not separately continuous.} We can now state the
following

\begin{thm}

Let $(\A,\Ao)$ be a quasi *-algebra with identity, $\delta$ a
*-derivation on $\Ao$ and $\pi$ a faithful *-representation of
$(\A,\Ao)$ such that the induced derivation $\delta_\pi$ is well
defined. Then, the following statements hold:

(1) Suppose that \be \forall \eta\in\D_\pi \, \exists c_\eta>0:
p_\eta(\delta_\pi(\pi(x)))\leq c_\eta p_\eta(\pi(x)), \quad
\forall x\in \Ao(\delta), \label{extra3}\en then
$\sum_{k=0}^\infty\frac{t^k}{k!}\delta^k(x)$ converges for all $t$
in the topology $\sigma_s$ to an element of
$\overline{\Ao(\delta)}^{\sigma_s}$ which we call $\alpha^t(x)$;
$\alpha^t$ can be extended to $\overline{\Ao(\delta)}^{\sigma_s}$.

(2) Suppose that, instead of (\ref{extra3}), the following
inequality holds $$\exists c>0: \, \forall \eta_1\in\D_\pi \,\,
\exists A_{\eta_1}>0,\, n\in\mathbb{N} \mbox{ and }
\eta_2\in\D_\pi:$$ \be p_{\eta_1}(\delta_\pi^k(\pi(x)))\leq
A_{\eta_1} c^k k! k^n p_{\eta_2}(\pi(x)), \quad \forall x\in
\Ao(\delta),\,\forall k\in\mathbb{N}_0, \label{extra4}\en  then
$\sum_{k=0}^\infty\frac{t^k}{k!}\delta^k(x)$ converges, for
$t<\frac{1}{c}$ in the topology $\sigma_s$ to an element of
$\overline{\Ao(\delta)}^{\sigma_s}$ which we call $\alpha^t(x)$;
$\alpha^t$ can be extended to $\overline{\Ao(\delta)}^{\sigma_s}$.
\label{theoremextra1}
\end{thm}

\begin{proof}
(1) Iterating equation (\ref{extra3}) we find that
$p_\eta(\delta_\pi^k(\pi(x)))\leq c_\eta^k p_\eta(\pi(x))$, for
all natural $k$. It is easy to prove now the following inequality,
for $x\in\Ao(\delta)$ and $t>0$:
$$
q_\eta(\sum_{k=0}^N\frac{t^k}{k!}\delta^k(x))\leq
\sum_{k=0}^N\frac{t^k}{k!}q_\eta(\delta^k(x))\leq
\sum_{k=0}^N\frac{(tc_\eta )^k}{k!}q_\eta(x)\rightarrow
e^{tc_\eta}q_\eta(x).$$ This proves the existence of
$\alpha^t(x)=\sigma_s-\lim_{N,\infty}\sum_{k=0}^N\frac{t^k}{k!}\delta^k(x)$
for all $t$. The extension of $\alpha^t$ to
$\overline{\Ao(\delta)}^{\sigma_s}$ by continuity is
straightforward.

(2) The proof is based on analogous estimates.

\end{proof}

\vspace{3mm}

{\bf Remarks:--} (1) The first remark is related to the different
conditions (\ref{extra3}) and (\ref{extra4}). The second condition
is much lighter, but the price we have to pay is that $\alpha^t$
can be defined only on a finite interval.

(2) Condition (\ref{extra3}) could be changed by requiring that
the seminorms on the left and the right side of the inequality are
not necessarily the same. In this case, however, we also have to
require that the constant $c_\eta$ is independent of $\eta$ and
that $\sup_{\varphi\in\D_\pi}p_\varphi(\pi(x))<\infty$.

\vspace{3mm}

\begin{cor} Under the general assumptions of Theorem
\ref{theoremextra1} we have:

(1) If condition (\ref{extra3}) is satisfied then $\alpha^t$ maps
$\overline{\Ao(\delta)}^{\sigma_s}$ into
$\overline{\Ao(\delta)}^{\sigma_s}$ and  \be
\alpha^{t+\tau}(x)=\alpha^t(\alpha^\tau(x)), \quad \forall t,
\tau, \, \forall x\in\Ao(\delta); \label{extra5}\en

(2) If condition (\ref{extra4}) is satisfied then $\alpha^t$ maps
$\overline{\Ao(\delta)}^{\sigma_s}$ into
$\overline{\Ao(\delta)}^{\sigma_s}$ for $t<\frac{1}{c}$ and  \be
\alpha^{t+\tau}(x)=\alpha^t(\alpha^\tau(x)), \quad \forall t,
\tau, \mbox{ with } t+\tau<\frac{1}{c}, \, \forall
x\in\Ao(\delta). \label{extra5bis}\en

\end{cor}

\begin{proof}

(1) The proof of the first statement is a trivial consequence of
the definition of $\alpha^t$ as given in Theorem
\ref{theoremextra1}: for all
$y\in\overline{\Ao(\delta)}^{\sigma_s}$ then
$\alpha^t(y)\in\overline{\Ao(\delta)}^{\sigma_s}$.

In order to prove equation (\ref{extra5}), we begin by fixing
$x\in\Ao(\delta)$. We have, for all $\eta\in\D_\pi$,
$$
q_\eta(\alpha^{t+\tau}(x)-\alpha^t(\alpha^\tau(x)))\leq
q_\eta(\alpha^{t+\tau}(x)-\alpha_N^{t+\tau}(x))+q_\eta(\alpha_N^{t+\tau}(x)-
\alpha_N^t(\alpha_N^\tau(x)))+$$ \be
+q_\eta(\alpha_N^t(\alpha_N^\tau(x))-\alpha^t(\alpha^\tau(x))),
\label{extra6}\en where
$\alpha_N^t(x)=\sum_{k=0}^N\frac{t^k}{k!}\delta^k(x)$. First we
observe that, because of Theorem \ref{theoremextra1},
$q_\eta(\alpha^{t+\tau}(x)-\alpha_N^{t+\tau}(x))\rightarrow 0$ for
all $t, \tau$ and $\eta\in \D_\pi$. The proof of the convergence
to zero of the third contribution,
$q_\eta(\alpha_N^t(\alpha_N^\tau(x))-\alpha^t(\alpha^\tau(x)))
\rightarrow 0$, follows from the fact that $\forall\,
\eta\in\D_\pi$ and $\forall\,\epsilon>0$ there exists
$N(\epsilon,\eta)>0$ such that, $\forall \,N>N(\epsilon,\eta)$,
$$
q_\eta(\alpha_N^t(\alpha_N^\tau(x))-\alpha^t(\alpha^\tau(x)))\leq
q_\eta((\alpha_N^t-\alpha^t)(\alpha_N^\tau(x)))+
q_\eta(\alpha^t(\alpha_N^\tau(x)-\alpha^\tau(x)))\leq$$
$$
\leq \sum_{k=N+1}^\infty
\frac{(tc_\eta)^k}{k!}q_\eta(\alpha_N^t(x))+\sum_{k=N+1}^\infty
\frac{\tau^k}{k!}q_\eta(\alpha^t(\delta^k(x)))\leq $$
$$ \leq
\sum_{k=N+1}^\infty
\frac{(tc_\eta)^k}{k!}(\epsilon+q_\eta(\alpha^t(x)))+q_\eta(x)e^{tc_\eta}\sum_{k=N+1}^\infty
\frac{(\tau c_\eta)^k}{k!}\rightarrow 0,
$$
as $N\rightarrow\infty$, for all $t, \tau$ (which we are assuming
to be positive here) and $x\in\Ao(\delta)$.

The conclusion follows from the fact that we also have \be
q_\eta(\alpha_N^{t+\tau}(x)-
\alpha_N^t(\alpha_N^\tau(x)))\rightarrow 0,\label{extra7}\en for
all $t, \tau$
 and $x\in\Ao(\delta)$. This can be proved by a direct estimate on
 the difference $\alpha_N^{t+\tau}(x)-
\alpha_N^t(\alpha_N^\tau(x))$ which can be written as
$\sum_{n=0}^N\sum_{l+k=n}A_{lk}-\sum_{l=0}^N\sum_{k=0}^NA_{lk}$,
where we have introduced
$A_{lk}=\frac{t^l\tau^k}{l!k|}\delta^{l+k}(x)$ for shortness.
Equation (\ref{extra7}) can now be proved using the same estimate
as for the third contribution.

The extension to $\overline{\Ao(\delta)}^{\sigma_s}$ is now
straightforward.

(2) The proof is only a minor modification of the one above.

\end{proof}

\vspace{3mm}

{\bf Remark:--} In general, for fixed $t$, $\alpha^t$ is not an
automorphism of $\overline{\Ao(\delta)}^{\sigma_s}$. This is
essentially due to the fact that the multiplication is not
continuous with respect to the topology $\sigma_s$.

\section{The case of proper CQ*-algebras}
{ As discussed in the Introduction, a standard assumption in the
algebraic approach to quantum systems is that the *-algebra $\Ao$
of local observables is a C*-algebra. For this reason, in this
Section, we will specialize our discussion to a particular class
of quasi*-algebras, named {\it proper CQ*-algebras}, that arise
when completing a C*-algebra $\Ao$ with respect to a weaker norm.
More precisely, a proper CQ*-algebra $(\A,\Ao)$ is constructed in
the following way: assume that $\Ao[\| \ \|_0]$ is a C*-algebra
and let $\| \ \|$ be another norm on $\Ao$ satisfying the
following two conditions:
\begin{itemize}
\item[(i)] $\|x^*\|= \|x\|, \quad , \forall x \in \Ao$
\item[(ii)]$\|xy\|\leq \|x\|_0\|y\|, \quad \forall x,y \in \Ao$.
\end{itemize}
Let $\A$ be the $\| \ \|$-completion of $\Ao$. The quasi *-algebra
$(\A,\Ao)$ is then a proper CQ*-algebra. For details we refer to
\cite{book}. We remark here that the construction outlined above
does not yield the most general type of proper CQ*-algebra, but it
produces the right object needed in our discussion.

The advantage of considering proper CQ*-algebras relies on the
fact that this makes easier to use some known results which hold
for bounded operators. This is convenient mainly because, as we
have seen in the previous section, the fact that the implementing
operator $H^{(\pi)}$ belongs to $\LL(\D_\pi,\D_\pi')$ makes it
impossible, in general, to consider powers of $H^{(\pi)}$. For
this reason we have proposed in Section 3 a different strategy,
which may appear rather peculiar. In this section we show that
some standard result can be used easily if we add an extra
assumption to the sesquilinear forms which produce (and are
produced by) the *-representations of CQ*-algebras we are going to
work with.

For proper CQ*-algebras Theorem 4.1 of \cite{bit} gives}
\begin{thm}
Let $(\A,\Ao)$ be a proper CQ*-algebra with unit and $\delta$ be a
*-derivation on $\Ao$. Then the following statements are
equivalent:

(i) There exists a positive sesquilinear form $\varphi$ on
$\A\times\A$ such that:

$\varphi$ is invariant, i.e.
\begin{equation}
\varphi(ax,y)=\varphi(x,a^*y), \mbox{ for all } a\in \A \mbox{ and
} x,y\in\Ao; \label{61}
\end{equation}

$\varphi$ is $\|\,.\|$-continuous, i.e.
\begin{equation}
|\varphi(a,b)|\leq \|a\| \|b\|, \mbox{ for all } a,b\in \A,
\label{62}
\end{equation}

$\varphi$ satisfies the following inequalities:
\begin{equation}
|\varphi(\delta(x),\1)|\leq
C(\sqrt{\varphi(x,x)}+\sqrt{\varphi(x^*,x^*)}), \quad \forall x\in
\Ao, \label{63}
\end{equation}
for some positive constant $C$, and $\forall a\in\A$ there exists
some positive constant $\gamma_a^2$ such that
\begin{equation}
|\varphi(ax,ax)|\leq \gamma_a^2\varphi(x,x), \quad \forall x\in
\Ao. \label{64}
\end{equation}

(ii) There exists a $(\|\,\|-\tau_s)$-continuous, ultra-cyclic and
bounded *-representation $\pi$ of $\A$, with ultra-cyclic vector
$\xi_0$, such that the *-derivation $\delta_\pi$ induced by  $\pi$
is s-spatial, i.e. there exists  a symmetric operator $\hat H$ on
the Hilbert space of the representation $\Hil_\pi$ such that
\be\left\{
\begin{array}{ll}
D(\hat H)=\pi(\Ao)\xi_0 \\
\delta_\pi(\pi(x))\Psi=i[\hat H,\pi(x)]\Psi,\quad \forall x\in\Ao,
\forall\Psi\in D(\hat H). \end{array} \right. \label{65} \en
\label{theorem61}
\end{thm}

The proof of this theorem is not significantly different from that
given in \cite{bit} and will be omitted here. It is worth
remarking that condition (\ref{64}) implies that the
representation $\pi_\varphi$, constructed starting from $\varphi$
as in \cite{bit}, is bounded, i.e. $\pi_\varphi(a)\in
B(\Hil_\varphi)$ for all elements $a\in\A$. However, since $\A$ is
not an algebra, $ab$ is not defined for general $a,b\in\A$, while
$\pi_\varphi(a) \pi_\varphi(b)$ turns out to be a bounded
operator. Therefore it has no meaning asking whether
$\pi_\varphi(ab)=\pi_\varphi(a) \pi_\varphi(b)$, since the lhs is
not well defined, unless $a$ and/or $b$ belongs to $\Ao$. In fact
in this case we can check that \be \pi_\varphi(ax)=\pi_\varphi(a)
\pi_\varphi(x),\quad \forall a\in\A,\,\forall x\in\Ao
\label{66}\en We also want to remark that the proof of the
implication $(i)\Rightarrow (ii)$ is mainly based on condition
(\ref{64}) which makes it possible to use the well known result
stated, for instance, in \cite{brarob}, Proposition 3.2.28.
Another remark which may be of some help in concrete applications
is the following: suppose that our sesquilinear form satisfies the
following modified version of (\ref{64}): $|\varphi(yx,yx)|\leq
\gamma^2\varphi(x,x)$, $\forall x,y\in \Ao$, with $\gamma$
independent of $y$. In this case, due to the norm-continuity of
$\varphi$, condition (\ref{64}) easily follows.

\medskip It is very easy to construct examples of positive
sesquilinear forms on  $\A\times\A$ satisfying (\ref{61}),
(\ref{62}) and (\ref{63}) above; in fact, the results in
\cite{btellepi} suggest to define, on the abelian proper
CQ*-algebra $(L^p(X,\mu), C(X))$ where $X$ a compact interval of
the real line, $\mu$ the Lebesgue measure and $p\geq 2$, a
sesquilinear form as, e.g.,
$\varphi(f,g):=\int_Xf(x)\overline{g(x)}\Psi(x)d\mu$, where we
take here $\Psi(x)=Ne^{\gamma x}$, and $N, \gamma>0$ are such that
$\|\Psi\|_{\frac{p}{p-2}}\leq 1$ (we put $\frac{p}{p-2}=\infty$ if
$p=2$). More difficult is to find examples of sesquilinear forms
satisfying also condition (\ref{64}). We refer to \cite{tra} for a
general analysis on this (and the other) requirements, while we
construct here an example in the context of Hilbert algebras,
which are relevant, e.g., in the Tomita-Takesaki theory.

Let $\A$ be an achieved Hilbert algebra with identity  and
$\Hil_\A$ the Hilbert space obtained by the completion of $\A$,
\cite{stsz,take}. For any $a\in\A$ we put $L_a^ob=ab$,
$R_a^ob=ba$, $b\in\A$. Then $L_a^o$ and $R_a^o$ can be extended to
bounded linear operators $L_a$ and $R_a$ on $\Hil_\A$,
respectively. The sets $L_\A$ and $R_\A$ are von Neumann algebras
on $\Hil_\A$, and $JL_aJ=R_{a^*}$ for all $a\in\A$, so that $JL_\A
J=R_\A$. Here $J$ is the isometric involution on $\Hil_\A$ which
extends the involution * of $\A$. Furthermore, for any
$x\in\Hil_\A$, we define two operators on $\A$ as $L_xa=R_ax$ and
$R_xa=L_ax$, for $a\in\A$ (we use the same symbol $L$ and $R$
since no confusion can arise). It is known, \cite{pallu}, that
$L_x$ and $R_x$ are closable operators and
$L_x^*=\overline{L}_{Jx}$, $LR_x^*=\overline{R}_{Jx}$, for all
$x\in\Hil_\A$.

It is also known that the Hilbert space $\Hil_\A$ over the
C*-algebra $\A$ with the norm $\|x\|_\flat=\|R_x\|$ ($\|\,\|$ is
the operator norm) and with the involution $J=*$ is a proper
CQ*-algebra, \cite{bithcq}. Here we consider a family
$\{\A_\lambda\}_{\lambda\in\Lambda}$ of Hilbert algebras. The
direct sum $\bigoplus_{\lambda\in\Lambda}\Hil_{\A_\lambda}$ of the
Hilbert spaces $\Hil_{\A_\lambda}$ is a proper CQ*-algebra under
the usual operations. Now we assume that one $\A_{\lambda_0}$ is a
H*-algebra, that is $\A_{\lambda_0}=\Hil_{\A_{\lambda_0}}$. Then
we consider a *-derivation $\delta$ of
$\bigoplus_{\lambda\in\Lambda}\A_\lambda$ satisfying $\delta
P_\lambda=P_\lambda\delta$ for all $\lambda\in\Lambda$, that is,
$\delta: \A_\lambda\rightarrow \Hil_{\A_{\lambda}}$, $\forall
\lambda\in\Lambda$, where $P_\lambda$ is the projection of
$\bigoplus_{\lambda\in\Lambda}\Hil_{\A_\lambda}$ onto
$\Hil_{\A_\lambda}$. If we finally define
$$\varphi_{P_\lambda}((x_\lambda),(y_\lambda))\equiv
<x_{\lambda_0},y_{\lambda_0}>, \quad \forall (x_\lambda),
(y_\lambda)\in\bigoplus_{\lambda\in\Lambda}\Hil_{\A_\lambda},$$
then $\varphi_{P_\lambda}$ satisfies all conditions required in
Theorem \ref{theorem61}.

\smallskip Once the symmetric operator $\hat H$ has been defined by
means of this Theorem, it is clear that if $\hat H$ is also
self-adjoint, then $e^{i\hat Ht}$ exists as a unitary operator in
$B(\Hil_\pi)$ and $\hat H$ is the generator of a one-parameter
group of unitary operators on $\Hil_\pi$.

\medskip

Let us now assume  that the representation of the proper
CQ*-algebra, $\pi$, satisfies all the requirement of (ii), Theorem
\ref{theorem61}, so to have an implementing , $\hat H$, operator
for $\delta_\pi$. Then we define the following set: \be
\Bo^\pi:=\{\pi(x)\in\pi(\Ao):\quad[\hat H,\pi(x)]_k\in\Ao,\quad
\forall k\in\mathbb{N}_0\},\label{extra8}\en which surely contains
all the elements $\lambda\pi(\1)$, $\lambda\in\mathbb{C}$. Also,
if $\hat H$ belongs to $\pi(\Ao)$, then $\Bo^\pi=\pi(\Ao)$. We
want to show now that $\Bo^\pi=\Ao(\delta)^\pi$. Using a simple
extension argument, we can first check that for all
$x\in\Ao(\delta)$  for which $\pi(x)\in\Bo^\pi$ we have \be
\delta_\pi(\pi(x))=\pi(\delta(x))=i[\hat H,\pi(x)].
\label{extra9}\en Also, it is evident that if $\pi(x)\in\Bo^\pi$,
then $[\hat H,\pi(x)]\in\Bo^\pi$. With this in mind we can now
prove the following

\begin{lemma}
$x\in\Ao(\delta)$ if and only if $\pi(x)\in\Bo^\pi$. For any such
element we have \be \delta_\pi^k(\pi(x))=i^k[\hat H,\pi(x)]_k,
\quad \forall k\in\mathbb{N}_0.\label{extra10}\en
\end{lemma}
\begin{proof}
Let us first take $x\in\Ao(\delta)$. Then (a)
$x\in\Ao\Rightarrow\pi(x)\in\pi(\Ao)$ and (b)
$\delta(x)\in\Ao\Rightarrow
\delta_\pi(\pi(x))=\pi(\delta(x))\in\pi(\Ao)$. Since $\pi$ is
bounded, this implies that $\delta_\pi(\pi(x))=i[\hat H,\pi(x)]$
and, as a consequence, that $[\hat H,\pi(x)]\in\pi(\Ao)$. The same
argument, applied  to $\delta(x)$ which, as we know, is still in
$\Ao(\delta)$, produces $\delta_\pi(\pi(\delta(x)))=i[\hat
H,\pi(\delta(x))]$ which, after few computations, produces
$\delta_\pi^2(\pi(x))=i^2[\hat H,\pi(x)]_2$ which still belongs to
$\pi(\Ao)$ since
$\delta_\pi^2(\pi(x))=\pi(\delta^2(x))\in\pi(\Ao)$. Iterating this
procedure we find that $\delta_\pi^k(\pi(x))=i^k[\hat H,\pi(x)]_k$
and $[\hat H,\pi(x)]_k\in\pi(\Ao)$ for all $k\in\mathbb{N}_0$.
Therefore $\pi(x)\in\Bo^\pi$.

In a similar way we can also prove that if $\pi(x)\in\Bo^\pi$ then
$x\in\Ao(\delta)$ and $\delta_\pi^k(\pi(x))=i^k[\hat H,\pi(x)]_k$
for all $k\in\mathbb{N}_0$.
\end{proof}

\vspace{3mm}

{\bf Remark:--} It may be worth recalling that  Lemma 4.2 also
implies that $\Ao(\delta)^\pi=\Bo^\pi$.

\vspace{3mm}

It is now straightforward to use this Lemma and Theorem
\ref{theoremextra1} to find conditions under which the sequence
$\alpha_N^t(x)=\sum_{k=0}^N\frac{t^k}{k!}\delta^k(x)$ is
$\sigma_s$-convergent and defines the time evolution of $x$,
$\alpha^t(x)$, for $x\in\Ao(\delta)$. Each one of the following
conditions can be used to deduce the existence of $\alpha^t(x)$:

{\bf Condition 1:--} $\forall \eta\in\D_\pi$ there exists
$c_\eta>0$ such that
$$
p_\eta([\hat H,\pi(x)])\leq c_\eta p_\eta(\pi(x))=c_\eta
q_\eta(x), \quad \forall x\in\Ao(\delta).
$$
In this case $\alpha^t(x)$ exists for all values of $t$.

\vspace{2mm}

{\bf Condition 2:--} $\exists c>0:\, \forall \eta_1\in\D_\pi$
there exist $A_{\eta_1}$, $n\in\mathbb{N}$ and $\eta_2\in\D_\pi$
such that
$$
p_{\eta_1}([\hat H,\pi(x)]_k)\leq A_{\eta_1}c^k k!k^n
q_{\eta_2}(x), \quad \forall x\in\Ao(\delta), \, \forall
k\in\mathbb{N}.
$$
In this case $\alpha^t(x)$ exists for all values of
$t<\frac{1}{c}$.

As we have already shown before, $\alpha^t$ can be extended to
$\overline{\Ao(\delta)}^{\sigma_s}$ and is a semigroup.

\vspace{3mm}

{ We adapt now Proposition 4.3 of \cite{bit} to the present
setting. For that, we again consider a family of *-derivations of
$\Ao$ and a single representation $\pi$ with the properties
required in (ii) of Theorem \ref{theorem61}: we remind that this
is the most common situation in physical applications. }

\begin{prop}

Let $\{\Sys_L, L\in\Lambda\}$ be a c-representable family such
that the corresponding c-representation $\pi$ is bounded . Also,
suppose that the following conditions hold:

(1) $\delta_n(x)$ is $\|\,\|$-Cauchy for all $x\in\Ao$;

(2) for all $n\in\mb{N}$ the induced *-derivation
$\delta_\pi^{(n)}$ is s-spatial, i.e. a symmetric operator $\hat
H_n$ on $\Hil_\pi$ exists such that \be\left\{
\begin{array}{ll}
D(\hat H_n)=\pi(\Ao)\xi_0 \\
\delta_\pi^{(n)}(\pi(x))\Psi=i[\hat H_n,\pi(x)]\Psi,\quad \forall
x\in\Ao, \forall\Psi\in D(\hat H_n). \end{array} \right.
\label{67} \en

(3) $\sup_n\|\hat H_n\xi_0\|=L<\infty$.

Then we have that

\begin{itemize}
\item[(a)] $\delta(x)=\|\,\|-\lim_n\delta_n(x)$ exists in $\A$ and
is a *-derivation of $\Ao$; \item[(b)] $\delta_\pi$, the
*-derivation induced by $\pi$, is well defined and s-spatial:
there exists a symmetric operator $\hat H$ on $\Hil_\pi$ such that
\be\left\{
\begin{array}{ll}
D(\hat H)=\pi(\Ao)\xi_0 \\
\delta_\pi(\pi(x))\Psi=i[\hat H,\pi(x)]\Psi,\quad \forall x\in\Ao,
\forall\Psi\in D(\hat H); \end{array} \right. \label{68} \en
\item[(c)] if $<\hat H_n\xi_0,\pi(y)\xi_0>\,\rightarrow\,<\hat
H\xi_0,\pi(y)\xi_0>$ for all $y\in\Ao$, then $<\hat
H_n\pi(x)\xi_0,\pi(y)\xi_0>\,\rightarrow\,<\hat
H\pi(x)\xi_0,\pi(y)\xi_0>$ for all $x,y\in\Ao$;\item[(d)] if
$\|(\hat H_n-\hat H)\xi_0\|\rightarrow 0$ for all $y\in\Ao$, then
$\|(\hat H_n-\hat H)\pi(x)\xi_0\|\rightarrow 0$ for all $x\in\Ao$.

\end{itemize}
\label{theorem63}
\end{prop}

\begin{proof}
The first three statements can be proven in quite the same way as
in \cite{bit}.

The proof of the statement $(d)$ is a consequence of the
definition of the implementing operator of a s-spatial derivation
as it can be deduced by Proposition 3.2.8 of \cite{brarob}. We
have
$$\left\{
\begin{array}{ll}
\hat H_n\pi(x)\xi_0=\frac{1}{i}\delta_\pi^{(n)}(\pi(x))\xi_0+\pi(x)\hat H_n\xi_0\:\mbox{ and } \\
\hat H\pi(x)\xi_0=\frac{1}{i}\delta_\pi(\pi(x))\xi_0+\pi(x)\hat
H\xi_0, \, \forall x\in\Ao \end{array} \right.$$ Therefore
$$\|(\hat H_n-\hat H)\pi(x)\xi_0\|\leq
\|(\delta_\pi^{(n)}(\pi(x))-\delta_\pi(\pi(x)))\xi_0\| + \|(\hat
H_n-\hat H)\xi_0\|\rightarrow 0$$ because of the assumptions on
$\hat H_n$ and $\pi$.

\end{proof}

\section*{Concluding remarks} As we have discussed in the Introduction, in this paper we
have chosen to regularize only the algebra related to a physical
system, leaving unchanged the set of states. However, in some
approaches discussed in the literature, see \cite{sew} for
instance, a {\it cutoff} is introduced for both the states and the
algebra. If we consider for a moment this point of view here, we
wonder what can be said if we have a family of positive
sesquilinear forms $\varphi_n$ on $\Ao\times\Ao$ instead of a
single one on $\A\times\A$. The simplest situation, which is the
only one we will consider here, is when the family
$\{\varphi_n,\,n\in\mathbb{N}\}$ satisfies the following
requirements:
\begin{itemize}
\item $\varphi_n(xy,z)=\varphi_n(y,x^*z)\quad \forall
x,y,z\in\Ao,\,\forall n\in\mb{N};$ \item $|\varphi_n(x,y)|\leq
\|x\|\|y\|\quad \forall x,y\in\Ao,\,\forall n\in\mb{N}$; \item the
sequence $\{\varphi_n(x,y)\}_{n\in\mb{N}}$ is Cauchy $\forall
x,y\in\Ao$.
\end{itemize}
The second condition allows us to extend each $\varphi_n$ to the
whole  $\A\times\A$. This extension, $\tilde\varphi_n$, satisfies
the same conditions as above. Furthermore, since
$\{\tilde\varphi_n(a,b)\}_{n}$ is a Cauchy sequence $\forall
a,b\in\A$, we can also define a new positive sesquilinear form
$\Phi$ on $\A\times\A$:
$$
\Phi(a,b)=\lim_{n}\tilde\varphi_n(a,b), \quad \forall a,b\in\A.
$$
It is clear that also $\Phi$ is *-invariant and
$\|\,\|$-continuous. If we also have that
\begin{itemize}
\item $|\tilde\varphi_n(\delta(x),\1)|\leq
C\sqrt{\varphi_n(x,x)+\varphi_n(x^*,x^*)}\quad \forall
x\in\Ao,\,\forall n\in\mb{N}$ and for some positive $C$; \item
$\forall a\in\A \,\; \exists\, \gamma_a>0:
|\tilde\varphi_n(ax,ax)|\leq \gamma_a^2\varphi_n(x,x) \quad\forall
x\in\Ao,\, \forall n\in\mb{N}$,
\end{itemize}
then we easily extend these properties to $\Phi$ so that we get a
positive sesquilinear form on $\A\times\A$ satisfying all the
requirements of point (i), Theorem \ref{theorem61}. Therefore we
have two different possibilities, at least if we are dealing with
a single derivation $\delta$:

{\em First possibility}: we use each $\tilde\varphi_n$ to
construct, using Theorem \ref{theorem61}, a *-representation
$\pi_n$ and an induced derivation
$\delta_{\pi_n}(\pi_n(x))=\pi_n(\delta(x))$, $x\in\Ao$, which
turns out to be s-spatial. Therefore we find a sequence of
symmetric operators $\hat H_n$ acting in possibly different
Hilbert spaces $\Hil_n$.

{\em Second possibility}: we use  $\Phi$ to construct, using again
Theorem \ref{theorem61}, a single *-representation $\pi$ and an
induced derivation $\delta_{\pi}(\pi(x))=\pi(\delta(x))$,
$x\in\Ao$, which is s-spatial. Therefore we get a symmetric
operator $\hat H$ acting on the Hilbert space of the
representation $\Hil$.

Both of these possibilities have a certain interest. We will
analyze in a forthcoming paper the details of these constructions
and the relations between $\hat H$ and $\hat H_n$.

\vspace{6mm} \noindent{\large \bf Acknowledgments} \vspace{5mm}

We acknowledge the financial support  of the Universit\`a degli
Studi di Palermo (Ufficio Relazioni Internazionali) and of the
Italian Ministry of Scientific Research. FB and CT wish to thank
all people at the Department of Applied Mathematics of  Fukuoka
University for their warm hospitality.

\end{document}